\documentclass[english,11pt,oneside]{article}
\usepackage{mymacros}
\usepackage{epstopdf}
\usepackage[a4paper, margin=2.5cm]{geometry}
\usepackage{natbib}
\usepackage[margin=1.0cm]{caption}
\setcitestyle{numbers,open={[},close={]}}

\newtheorem{theorem}{Theorem}
\newtheorem{lemma}{Lemma}
\newtheorem{proposition}{Proposition}
\newtheorem{cor}{Corollary}

\newtheorem{assumption}{Assumption}
\newtheorem{remark}{Remark}

\newtheorem{definition}{Definition}

\newtheorem{property}{Property}

\makeatletter
\newcommand{\leqnomode}{\tagsleft@true}
\newcommand{\reqnomode}{\tagsleft@false}
\makeatother

\makeatletter
\let\@fnsymbol\@arabic
\makeatother

\newenvironment{proof}{\noindent \textbf{Proof.}}{\hfill $\blacksquare$}
\DeclareMathOperator*{\argmin}{arg\,min}

\title{Computing the channel capacity of a communication system affected by uncertain transition probabilities}
\author{Krzysztof Postek \thanks{Faculty of Industrial Engineering and Management, Technion - Israel Institute of Technology, e-mail: krzysztofp@technion.ac.il} 
\and Aharon Ben-Tal \thanks{Faculty of Industrial Engineering and Management, Technion - Israel Institute of Technology; Shenkar College, Israel, e-mail: abental@ie.technion.ac.il}}
\date{}
\begin{document}

\setlength{\parindent}{0pt}
\setlength{\parskip}{5pt}
\maketitle
\abstract{We study the problem of computing the capacity of a discrete memoryless channel under uncertainty affecting the channel law matrix, and possibly with a constraint on the average cost of the input distribution. The problem has been formulated in the literature as a max-min problem. We use the robust optimization methodology to convert the max-min problem to a standard convex optimization problem. For small-sized problems, and for many types of uncertainty, such a problem can be solved in principle using interior point methods (IPM). However, for large-scale problems, IPM are not practical. Here, we suggest an $\mathcal{O}(1/T)$ first-order algorithm based on \cite{Nemirovski2004paper} which is applied directly to the max-min problem.}

\section{Introduction} \label{sec.introduction}
Channel capacity is a fundamental notion in the field of information theory that started with the seminal paper of \cite{Shannon1948}. For a given channel, it provides an upper bound on the rate at which information can be reliably transmitted. In this paper we consider the problem of computing the capacity of a discrete memoryless channel (DMC). A DMC consists of an input alphabet of $N$ symbols, each used with probability $p_n$ determined by the user, and an output alphabet $Y$ of $M$ symbols. Given the $n$-th input symbol, the $m$-th output symbol occurs with a conditional probability $Q_{nm}$, where the matrix $Q \in \mathbb{R}^{N \times M}$ is known as the \textit{channel law matrix}. The channel capacity of a DMC is given by:
\leqnomode
\begin{equation}
\sup\limits_{p \in \Delta_N} I(p,Q) := \sup\limits_{p \in \Delta_m} \sum\limits_{n=1}^N \sum\limits_{m=1}^M p_n Q_{nm} \log \frac{Q_{nm}}{\sum\limits_{l=1}^N p_l Q_{lm}}, \label{equation.cc} \tag{CC}
\end{equation}
\reqnomode
where the maximized term is known as the \textit{average mutual information} (AMIM). The AMIM is a convex function of the input distribution $p$ and hence, at least in theory, convex optimization algorithms can be used to solve (\ref{equation.cc}). However, the performance of such algorithms can be poor and the need for other approaches emerged. One of the first algorithms is the well-known Arimoto-Blahut algorithm (see \cite{Arimoto1972,Blahut1972}). Later, more efficient algorithms have been developed to solve the capacity problem and some of its special cases. A comprehensive survey of algorithms for solving (\ref{equation.cc}) is given in \cite{Sutter2015}. This paper also suggests a first-order algorithm based on \cite{Nesterov2005} for problem (\ref{equation.cc}) with an additional linear constraint. 


In reality, channels correspond to physical devices affected by inaccuracies and/or whose parameters are measured only up to a certain accuracy. Therefore, the channel at hand may differ from its assumed version. Already in the 1950's researchers were considering the possibility that the channel is not known fully. This case, known as \textit{compound channel} case, has a rich literature, see e.g. \cite{Blackwell1959,Dobrushin1959,Wolfowitz1959}. A comprehensive survey by \cite{Lapidoth1998} summarizes these results. In this paper, the uncertain channel capacity problem is represented as:
\leqnomode
\begin{equation}
\sup\limits_{p \in \Delta_m} \inf\limits_{Q \in \mathcal{Q}} I(p,Q) , \label{equation.cc.rc} \tag{R-CC}
\end{equation}
\reqnomode
where the channel law matrix $Q$ is only known to belong to a set $\mathcal{Q}$. In today's language (\ref{equation.cc.rc}) could be called a robust optimization problem \cite{BenTal2009,BenTal2015} where $\mathcal{Q}$ is the so-called \textit{uncertainty set}. To the best of our knowledge, there are no papers in the literature on algorithms for solving (\ref{equation.cc.rc}). The main purpose of our paper is to fill this gap by designing an efficient algorithm for solving large scale instances. Contrary to \cite{Sutter2015}, we cannot use the algorithm of \cite{Nesterov2005} since the coupling term in (\ref{equation.cc.rc}) is not bilinear. Instead, we propose a first-order algorithm based on \cite{Nemirovski2004paper}. The performance of this algorithm relies on two types of operation that need to be executed in each iteration. The first of them is the computation of the gradient of the AMIM with respect to $p$ and $Q$. Secondly, at each iteration the algorithm calls an oracle consisting of a proximal minimization. We show that for our problem, an analytic expression for the minimum is obtained, or the proximal minimization reduces to solving a single-variable convex problem.


The contribution of our paper is as follows:
\begin{itemize}
\item we show that the robust counterpart (RC) of the min-max problem (\ref{equation.cc.rc}) reduces to a standard deterministic convex optimization problem. Moreover, we derive a dual form of RC. For the binary symmetric channel, this reduces to an explicit formula and for a weakly symmetric channel with Kullback-Leibler uncertainty it reduces to a single-variable problem;
\item we adapt the algorithm of \cite{Nemirovski2004paper} to (\ref{equation.cc.rc}) and derive simple conditions under which the $\mathcal{O}(1/T)$ convergence can be proved;
\item similar to \citet{Sutter2015}, we look at an extended version of problem (\ref{equation.cc.rc}) which includes an average cost constraint, and we show how to extend the algorithm to handle this case;
\item by numerical experiments, we illustrate the effect of uncertainty on the channel capacity.
\end{itemize}
The organization of the paper is as follows. In Section \ref{sec.derivations} we derive the convex formulation of the uncertain channel capacity problem and its dual. In Section~\ref{subsec.Nemirovski} we adapt the algorithm of \cite{Nemirovski2004paper} to solve the uncertain channel capacity problem for large-scale instances. Section~\ref{sec.experiments} includes three numerical experiments illustrating the impact of uncertainty on channel capacity.

\textbf{Notation.} We denote by $\Delta_{N}$ the $N$-dimensional simplex and by $\Delta_{N \times M}$ is a set of matrices where each row belongs to $\Delta_M$. $C^{1,1}(A)$ shall denote the space of differentiable functions with a Lipschitz continuous gradient over the set $A$. By $\partial g(\lambda)$ we denote the subdifferential of function $g(\cdot)$ at point $\lambda$. For a norm $\| \cdot \|$ we denote by $\| \cdot \|_*$ its dual norm.
\section{The uncertain channel capacity problem} \label{sec.derivations}
\subsection{Deriving the robust counterpart (RC)}
The RC is derived by dualizing the inner problem in the min-max (\ref{equation.cc.rc}) (Proposition~\ref{proposition.dual.counterpart}). This results in a standard maximization of a concave function over a convex set.
\begin{proposition} \label{proposition.dual.counterpart}
Let the uncertainty set $\mathcal{Q} \subset \mathbb{R}^{N \times M}$ be a compact convex set. Then, the worst-case capacity of a channel is given by the optimal value to the following problem:
\begin{align}
\max_{p,\lambda, V} & \ \sum\limits_{n=1}^N \lambda_n - \delta^\ast (V | \mathcal{Q}) \label{eq.cc.rc} \tag{P} \\
\text{\emph{s.t.}} & \ \sum\limits_{n=1}^N p_n \exp \left( \frac{\lambda_n - V_{nm}}{p_n}  \right) \leq 1, \ \forall m \nonumber \\
& \ p \geq 0  \nonumber \\
& \ \sum\limits_{n=1}^N p_n = 1, \nonumber 
\end{align}
where $\delta^\ast (V | \mathcal{Q})$ is the support function of the set $\mathcal{Q}$ evaluated at $V \in \mathbb{R}^{N \times M}$.
%
%
\end{proposition}
\begin{proof}
We begin by deriving the dual of the inner problem in (\ref{equation.cc.rc}). Let:
$$
f(Q,p) = \sum\limits_{n=1}^N \sum\limits_{m=1}^M p_n Q_{nm} \log \frac{Q_{nm}}{\sum\limits_{l=1}^N p_l Q_{lm}}.
$$
Using the results of \cite{BenTal2015}, relying on Fenchel duality, we have:
\begin{align}
\inf\limits_{Q \in \mathcal{Q}} f(Q,p) & = \inf\limits_{Q} \left\{  f(Q,p) + \delta \left( Q | \mathcal{Q} \right)\right\} \nonumber \\
& = \inf\limits_{Q} \left\{  \delta \left( Q | \mathcal{Q} \right)  - ( - f(Q,p) ) \right\} \nonumber \\
& = \sup\limits_{V} \left\{  (-f)_\ast(V,p) - \delta^\ast \left( V | \mathcal{Q} \right) \right\}, \label{eq.Fenchel.dualized}
\end{align}
where $(-f)_\ast(V,p)$ is the concave conjugate of the function $-f$ with respect to the first variable and $\delta^\ast \left( V | \mathcal{Q} \right)$ is the support function of the set $\mathcal{Q}$, evaluated at $V \in \mathbb{R}^{N \times M}$. Therefore, we need to find:
\begin{align}
(-f)_\ast(V,p) = \inf\limits_{Q \in \Delta_{N \times M}} \sum\limits_{n=1}^N \sum\limits_{m=1}^M \left( p_n Q_{nm} \log \frac{Q_{nm}}{\sum\limits_{l=1}^N p_l Q_{lm}} + V_{nm} Q_{nm} \right) \label{eq.concave.conjugate}
\end{align}
So, let us define new variables:
$$
x_{nm} = p_n Q_{nm}, \ \forall n,m, \qquad q_m = \sum_{n=1}^N p_n Q_{nm}
$$
in terms of which
\begin{align*}
\inf\limits_{Q \in \Delta_{N,M}} \sum\limits_{n=1}^N \sum\limits_{m=1}^M \left( p_n Q_{nm} \log \frac{Q_{nm}}{\sum\limits_{l=1}^N p_l Q_{lm}} + V_{nm} Q_{nm} \right)
\end{align*}
is equivalent to
\begin{align*}
\inf\limits_{x_{nm}, q_m} & \ - \sum\limits_{n=1}^N p_n \log p_n + \sum\limits_{n=1}^N \left( \sum\limits_{m=1}^M x_{nm} \log \frac{x_{nm}}{q_m} + V_{nm} \frac{x_{nm}}{p_n} \right) \\
\text{s.t.} & \ \sum\limits_{m=1}^M x_{nm} = p_n \quad \forall n \\
& \ q \in \Delta_M \\
& \ x_{nm} \geq 0 \quad \forall n,m. \\
\end{align*}
Next, we use Lagrangian duality to compute the dual of the latter problem :
\begin{align*}
& \inf_{q \in \Delta_M, x_{nm} \geq 0} \sum\limits_{n=1}^N \sum\limits_{m=1}^M \left( x_{nm} \log \frac{x_{nm}}{q_m}  + V_{nm} Q_{nm} \right) - \sum\limits_{n=1}^N u_n \left( \sum\limits_m x_{nm} - p_n \right) \\
& =
\sum\limits_{n=1}^N u_n p_n + \inf_{q \in \Delta_n} \inf\limits_{x_{nm} \geq 0} \sum\limits_{n=1}^N \sum\limits_{m=1}^M q_m \left( \frac{x_{nm}}{q_m} \log \frac{x_{nm}}{q_m} - \frac{x_{nm}}{q_m} \left( u_n - \frac{V_{nm}}{p_n} \right)  \right) \\
& =
\sum\limits_{n=1}^N u_n p_n + \inf_{q \in \Delta_n} \sum\limits_{n=1}^N \sum\limits_{m=1}^M q_m\exp \left( u_n - \frac{V_{nm}}{p_n} - 1 \right)  \\
& =
\sum\limits_{n=1}^N u_n p_n - \max\limits_m  \sum\limits_{n=1}^N\exp \left( u_n - \frac{V_{nm}}{p_n} - 1 \right)  
\end{align*}
Introducing an auxiliary variable $\mu$ for the $\max$ term and using strong duality, we obtain:
\begin{align*}
& &(-f)_\ast(V,p)= \max_{u_n, \mu} &\ - \sum\limits_{n=1}^N p_n \log p_n + \sum\limits_{n=1}^N u_n p_n - \mu \\
&& \text{s.t.} & \ \mu \geq \sum\limits_{n=1}^N\exp \left( u_n - \frac{V_{nm}}{p_n} - 1 \right) && \forall m.
\end{align*}
We introduce auxiliary variable $\lambda_n$ such that $u_n = \log p_n + \lambda_n / p_n + \mu$ to obtain
\begin{align*}
& &(-f)_\ast(V,p)= \max_{\lambda, \mu} & \ \sum\limits_{n=1}^N \lambda_n \\
&& \text{s.t.} & \ \mu \geq \sum\limits_{n=1}^N p_n \exp \left( \frac{\lambda_n - V_{nm}}{p_n} - 1 + \mu \right) && \forall m.
\end{align*}
Each of the constraints has a form $\mu \geq A \exp(\mu)$ where $A > 0$. The property of such a constraint is:
$$
A \geq A' > 0 \quad \Rightarrow \quad \left( \mu \geq A \exp(\mu) \ \Rightarrow \ \mu \geq A' \exp(\mu) \right).
$$
Therefore, if there exists a $\mu$ satisfying one of the constraints for which the expression $\sum_{n=1}^N p_n \exp \left( (\lambda_n - V_{nm}) /p_n - 1\right)$ is smallest, this $\mu$ satisfies all of the constraints. Therefore:
\begin{align*}
\mu \geq \sum\limits_{n=1}^N p_n \exp \left( \frac{\lambda_n - V_{nm}}{p_n} - 1 + \mu \right) \ & \forall m && \Leftrightarrow && \mu_m \geq \sum\limits_{n=1}^N p_n \exp \left( \frac{\lambda_n - V_{nm}}{p_n} - 1 + \mu_m \right) & \forall m
\end{align*}
and one can eliminate $\mu_m$ from each inequality by moving the terms to the left-hand side and maximizing w.r.t. $\mu_m$. In the end, we obtain:
\begin{align*}
& &(-f)_\ast(V,p) = \max_{\lambda} & \sum\limits_{n=1}^N \lambda_n \\
&& \text{s.t.} & \sum\limits_{n=1}^N p_n \exp \left( \frac{\lambda_n - V_{nm}}{p_n}  \right) \leq 1 && \forall m.
\end{align*}
This expression is inserted back into (\ref{eq.Fenchel.dualized}), which in turn inserted into (\ref{equation.cc.rc}) results in the final form (\ref{eq.cc.rc}).
\end{proof}

Problem (\ref{eq.cc.rc}) is convex -- the objective is a concave function of the decision variables and the first constraint involves the perspectives of the (convex) exponential function. Thus if the support function $\delta^\ast(\cdot | \mathcal{Q})$ has a tractable form (see e.g. \cite{BenTal2015} for a review of computationally tractable uncertainty sets), it can be solved using convex optimization methods, e.g., the IPM. In the next section, we derive the problem dual to (\ref{eq.cc.rc}) which allows us to obtain upper bounds on the robust channel capacity.
\begin{remark}
In \cite{Sutter2015}, the deterministic channel capacity problem includes a linear average cost constraint. It models the case where use of the $n$-th symbol of the input alphabet incurs a cost $a_n$ and the aim is to keep the average cost below a threshold $b$. This gives a constraint $a^\top p \leq b$, where $a \in \mathbb{R}^{N}_{+}$ and $b \in \mathbb{R}_{+}$. In our robust setting, such a constraint is also easily included as it is sufficient to append it to problem (\ref{eq.cc.rc}).
\end{remark}
\subsection{Dual problem}
The following proposition gives the convex optimization problem dual to (\ref{eq.cc.rc}).\leqnomode
\begin{proposition} \label{proposition.dual.robust.counterpart}
A dual of problem (\ref{eq.cc.rc}) for which strong duality holds is given by:
\begin{align}
 \min_{v,Q} \ & \log \sum\limits_{m=1}^M \exp(v_m) + \max\limits_n \left\{ \sum\limits_{m=1}^M Q_{nm} (\log Q_{nm} - v_m ) \right\} \label{eq.cc.rc.dual} \tag{D} \\
\text{\emph{s.t.}} \ & Q \in \mathcal{Q}. \nonumber
\end{align}
\end{proposition}
\begin{proof}
See Appendix~\ref{appendix.proofs.1}.
\end{proof}
\reqnomode

The following corollary gives an upper bound on the robust channel capacity by inserting the feasible solution $v_m = \log \sum_{l} Q_{lm}$.
\begin{cor} \label{corollary.symmetric.channel}
The following is an upper bound on the robust channel capacity (\ref{equation.cc.rc}):
\begin{align}
\min (\ref{eq.cc.rc.dual}) \leq \min_{Q \in \mathcal{Q}} \ & \log N + \max\limits_n \left\{ \sum\limits_{m=1}^M Q_{nm} \log \frac{Q_{nm}}{\sum\limits_l Q_{lm}} \right\} \label{eq.upper.bound.symmetric}
\end{align}
\end{cor}
As it turns out, the upper bound of Corollary \ref{corollary.symmetric.channel} is tight for the class of \textit{weakly symmetric} channels. 
\begin{definition}
A channel is \textbf{weakly symmetric} if each row contains the same set of values as any other row, with possible permutations, and the column sums are equal.
\end{definition}
For example, the following channel is weakly symmetric:
$$
Q = \left[ \begin{array}{ccc} \frac{1}{3} & \frac{1}{6} & \frac{1}{2} \\ \frac{1}{3} & \frac{1}{2} & \frac{1}{6} \end{array} \right]
$$
and the following uncertainty set involves only weakly symmetric channels:
$$
\mathcal{Q} = \text{conv} \left( \left[ \begin{array}{ccc} \frac{1}{3} & \frac{1}{6} & \frac{1}{2} \\ \frac{1}{3} & \frac{1}{2} & \frac{1}{6} \end{array} \right] , \left[ \begin{array}{ccc} \frac{1}{3} & \frac{1}{2} & \frac{1}{6} \\ \frac{1}{3} & \frac{1}{6} & \frac{1}{2} \end{array} \right] \right).
$$
For weakly symmetric channels, it is sufficient to specify an uncertainty set for $Q$ by specifying the uncertainty set for a single row (with permuted values of the same entries in other rows) with the condition that the sums of the column entries stays equal.
\begin{proposition} \label{proposition.weakly.symmetric.bound.tight}
For weakly symmetric channels the bound (\ref{eq.upper.bound.symmetric}) is tight:
\begin{align}
\label{eq.upper.bound.symmetric.min}
\min (\ref{eq.cc.rc.dual})  = \min\limits_{Q \in \mathcal{Q}} \log N + \sum\limits_{m=1}^M Q_{1m} \log \frac{Q_{1m}}{\sum\limits_l Q_{lm}}.
\end{align}
\end{proposition}
\begin{proof}
See Appendix~\ref{appendix.proofs.2}.
\end{proof}
\subsection{Special cases} \label{sec.special.cases}
As an illustration of the results, we consider two simple cases: (i) the well-known binary symmetric channel, and (ii) the weakly symmetric channel under Kullback-Leibler uncertainty about the transition probabilities in each row of $Q$.
\subsubsection{Binary symmetric channel} \label{subsection.bsc}
A binary symmetric channel is one where $m = n = 2$ and 
$$
Q = \left[ \begin{array}{cc} 1 - \beta & \beta \\ \beta & 1- \beta \end{array} \right].
$$
Clearly, this channel is characterized by a single parameter $\beta$ and hence, the only possible compact convex uncertainty set is an interval $\beta \in [\underline{\beta}, \overline{\beta}]$. Then, the following result holds:
\begin{proposition} \label{proposition.bsc}
Assume w.l.o.g. that $\underline{\beta} \leq 1/2$. Then, the robust capacity of a binary symmetric channel is given by 
\begin{align*}
\log 2 + \beta^\ast \log \beta^\ast + (1-\beta^\ast) \log (1-\beta^\ast) \text{ where } \beta^\ast = \min\{ 1/2, \overline{\beta} \}.
\end{align*}
\end{proposition}
\begin{proof}
See Appendix~\ref{appendix.proofs.3}.
\end{proof}
\subsubsection{Weakly symmetric channels under Kullback-Leibler (KL) uncertainty}
Weak symmetry allows us to focus on a single row, so let us define $r_m = Q_{1m}$ and $\sum_{l=1}^N Q_{lm} = N/M$ for all $m$. Under KL uncertainty, the uncertainty set for the first row is given as:
$$
\mathcal{Q} = \left\{ r \in \mathbb{R}^M_{+}: \ 1^\top r = 1, \  \sum\limits_{m=1}^M r_m \log \left( \frac{r_m}{q_m} \right) \leq \rho \right\}.
$$ 
In this setting, the dual problem (\ref{eq.cc.rc.dual}) reduces to the following one:
\begin{align*}
\min\limits_{r_n} \ & \log N + \sum\limits_{m=1}^M r_m \log \frac{r_m}{N/M}  && & \min\limits_{r_n} \ & \log M + \sum\limits_{m=1}^M  r_m \log r_m \\
\text{s.t.} \ & \sum\limits_{m=1}^M r_m \log \left( \frac{r_m}{q_m} \right) \leq \rho & \Leftrightarrow &&  \text{s.t.} \ & \sum\limits_{m=1}^M r_m \log \left( \frac{r_m}{q_m} \right) \leq \rho \\
& \sum\limits_{m=1}^M r_m = 1 & &&  &\sum\limits_{m=1}^M r_m = 1  \\
& r_m \geq 0 && &&  r_m \geq 0.
\end{align*}
Omitting the first term $\log M$, we write the Lagrangian and minimize w.r.t. $r_m$ to obtain the dual function:
\begin{align*}
g(\mu,\lambda) \ & = \inf\limits_{r_m \geq 0} \sum\limits_{m=1}^M  r_m \log r_m + \lambda \left( \sum\limits_{m=1}^M r_m \log \left( \frac{r_m}{q_m} \right) - \rho \right) + \mu \left(  \sum\limits_{m=1}^M r_m - 1 \right) \\
 = & \inf\limits_{r_m \geq 0} \sum\limits_{m=1}^M (1+\lambda) r_m \log r_m - (\lambda \log q_m - \mu) r_m  - \lambda \rho - \mu\\
& \text{ which by optimizing w.r.t. } r_m \text{ gives} \\
 = & - \sum\limits_{m=1}^M (1+\lambda) \exp \left( \frac{\lambda \log q_m - \mu}{1 + \lambda} - 1 \right) - \lambda \rho - \mu, \\
\end{align*}
where $\lambda \geq 0$. To solve the dual problem $\max_{\mu, \lambda \geq 0} g(\mu,\lambda)$ we maximize first w.r.t. $\mu$ to obtain the optimality condition:
$$
\frac{d g(\mu,\lambda)}{d \mu} = 0 \ \Leftrightarrow \ - \sum\limits_{m=1}^M \exp \left( \frac{\lambda \log q_m - \mu}{1 + \lambda} - 1 \right) + 1 = 0 \ \Leftrightarrow \ \mu = (1 + \lambda) \log \left(  \sum\limits_{m=1}^M \exp \left( \frac{\lambda \log q_m }{1 + \lambda} - 1 \right) \right),
$$
which, after inserting it back, gives the following maximization problem over a single variable $\lambda \geq 0$:
\begin{align*}
\max\limits_{\lambda \geq 0} \ & \left\{ - \sum\limits_{m=1}^M (1+\lambda) \exp \left( \frac{ \lambda \log q_m - (1 + \lambda) \log \left(  \sum\limits_{m=1}^M \exp \left( \frac{ \lambda \log q_m }{1 + \lambda} -1 \right) \right)}{1 + \lambda} \right) \right. \\
& \left. - \lambda \rho - (1 + \lambda) \log \left(  \sum\limits_{m=1}^M \exp \left( \frac{ \lambda \log q_m }{1 + \lambda} - 1 \right) \right) \right\} . \nonumber
\end{align*}

\section{A proximal algorithm for solving the robust channel capacity problem (\ref{equation.cc.rc})}  \label{subsec.Nemirovski}
\subsection{Introduction} \label{sec.Nemirovski.introduction}
In principle, the robust channel capacity problem of Proposition~\ref{proposition.dual.counterpart} can be solved using the interior point methods for a wide range of convex and compact uncertainty sets $\mathcal{Q}$. However, the numerical performance of these methods is poor already on medium-sized instances which is in line with the experience already gained in the deterministic case \cite{Sutter2015}. This is particularly due to the presence of the variable matrix $V \in \mathbb{R}^{N \times M}$, which greatly increases the computational effort since the IPMs in each iteration require solving a linear system of size corresponding to the number of variables. This situation calls for using first-order methods. 

It turns out that in a suitable uncertainty setting, the original max-min problem (\ref{equation.cc.rc}) can be solved using the prox-algorithm of \cite{Nemirovski2004paper}. The algorithm has a convergence rate of $\mathcal{O}(1/T)$ and its performance relies on two types of operation that need to be executed in each iteration: (i) computing the gradient of the AMIM with respect to $p$ and $Q$; (ii) calling an oracle which performs a proximal minimization step. For the max-min problem (\ref{equation.cc.rc}) the partial derivatives of the AMIM with respect to $p$ and $Q$ are easily computed. As for the proximal minimization, an \emph{analytic expression for the minimum is obtained, or the proximal minimization reduces to solving a single-variable convex problem}.

In Appendix~\ref{appendix.Nemirovski} we state a self-contained description of the algorithm of \cite{Nemirovski2004paper}. Here, we first present our modelling setup and show that it satisfies the crucial assumptions for proving the convergence of the algorithm. In Section~\ref{subsec.prox.algorithm} we outline the algorithm and state its convergence. In Section~\ref{subsec.avg.cost.constraint} we extend the algorithm to the case with the average cost constraint added.

We propose a rather general modeling setup where the uncertainty set for the channel law matrix $Q$ is given as:
$$
\mathcal{Q} = \left\{Q(\xi): \ Q(\xi) = Q^0 + \sum\limits_{s=1}^S \xi_s Q^s, \ \mathbf{\xi} \in \mathcal{B} \right\},
$$
where $\mathcal{B} \subset \mathbb{R}^S$ is a `simple' compact convex set of perturbations $\xi$. By `simple' it is meant that it is easy to perform the proximal minimization step. Examples of such simple $\mathcal{B}$ are: $\infty$-norm ball, 2-norm ball, and simplex. For the first two cases, we add to the description of the $\mathcal{Q}$ constraints $Q^s \mathbf{1} = \mathbf{0}$ for all $s = 1,\ldots,S$, necessary for the row sums of $Q(\xi)$ to stay equal to $1$ for all $\xi \in \mathcal{B}$.

In this setup, the problem to solve is:
\begin{align} \label{equation.saddle.crude}
\max\limits_{p \in \Delta_N} \min\limits_{\xi \in \mathcal{B}} \phi(\xi,p) = \min\limits_{\xi \in \mathcal{B} } \max\limits_{p \in \Delta_N} \phi(\xi,p),
\end{align}
where
$$
\phi(\xi,p) = \sum\limits_{n=1}^N \sum\limits_{m=1}^M p_n Q_{nm}(\xi) \log \frac{Q_{nm}(\xi)}{\sum\limits_{l=1}^N p_l Q_{lm}(\xi)}
$$
and where the fact that the max-min $=$ min-max holds due to the Sion-Kakutani minimax theorem. We set $z = (\xi,p)$ to shorten the notation and define the gradient mapping
$$
\Phi(z) = \left[ \begin{array}{c} \nabla_\xi \phi(z) \\ - \nabla_p \phi(z) \end{array} \right].
$$
The three crucial assumptions for the algorithm's convergence are: (i) compactness of the domain, (ii) Lipschitz continuity of $\Phi(z)$, (iii) existence of strongly convex distance-generating functions that it is easy to minimize them over the domain.

The compactness assumption (i) is satisfied due to the compactness of $\mathcal{B}$ and $\Delta_N$. Assumption (ii) can be easily met due to the following result.
\begin{proposition} \label{proposition.Nemirovski.L.continuity}
If for some $\tau > 0$ it holds that $Q_{nm} > \tau$ for all $Q \in \mathcal{Q}$, then the mapping:
$$
\Psi(Q,p) = \left[ \begin{array}{c} \nabla_Q f(Q,p) \\ \nabla_p f(Q,p)  \end{array} \right]
$$
is Lipschitz continuous on $\mathcal{Q} \times \Delta_N$. \hfill $\square$
\end{proposition}
\begin{proof}
See Appendix~\ref{appendix.proofs.4}.
\end{proof}
\begin{remark}
Note the extra assumption $\tau > 0$. A similar assumption is made in the deterministic channel case of \cite{Sutter2015}, in order to force the set of dual solutions of a certain problem to be compact, on which the convergence of their algorithm rests.
\end{remark}
The third assumption (iii) is the existence of a distance-generating function $\omega: \mathbb{R}^n \times \mathbb{R}_+^m \rightarrow \mathbb{R}$ such that the computation of the proximal operator:
$$
\text{Prox}_{z}(d) = \argmin\limits_{w \in \mathcal{B} \times \Delta_N} \left[ \omega(w) + \left\langle w, d  - \omega'(z)  \right\rangle  \right].
$$
is simple. In line with \cite{Nemirovski2004paper}, we use
\begin{align} \label{equation.dgf}
\omega(z) = \gamma_1 \omega_1(\xi)  + \gamma_2 \omega_2(p)
\end{align}
where $\omega_1(\cdot)$ and $\omega_2(\cdot)$ are strongly convex distance-generating functions over $\mathcal{B}$ and $\Delta_N$ and $\gamma_1$, $\gamma_2$ are suitably chosen positive constants. Possible choices for $\omega_1(\cdot)$ and $\omega_2(\cdot)$ are:
\begin{itemize}
\item $\omega_1(\cdot)$: 
\item[] for $\mathcal{B}$ -- $\infty$-norm ball or $2$-norm ball:
$$
\omega_1(\xi) = \frac{\xi^\top \xi}{2},
$$
\item[] for $\mathcal{B}$ -- simplex:
$$
\omega_1(\xi) = \sum\limits_{s=1}^S (\xi_s  + \delta / S) \log (\xi_s  + \delta / S),
$$
where $\delta > 0 $ is an arbitrarily positive number.
\item $\omega_2(\cdot)$:
$$
\omega_2(p) = \sum\limits_{n=1}^N (p_n  + \delta / N) \log (p_n  + \delta / N).
$$
\end{itemize}
With such a choice, the computation of the prox-operator is equivalent to:
\begin{align}\label{eq.Nemirovski.subproblem.solution}
& \min\limits_{w \in \mathcal{B} \times \Delta_N} \left[ \omega(w) + \left\langle w, d  - \omega'(z)  \right\rangle  \right] \\
&= \min\limits_{\xi \in \mathcal{B}, p \in \Delta_N} \left[ \gamma_1 \omega_1(\xi)  + \gamma_2 \omega_2(p) + \left\langle \left[ \begin{array}{c} \xi \\ p \end{array} \right], \left[ \begin{array}{c} d_1 \\ d_2 \end{array} \right]  - \left[ \begin{array}{c} \gamma_1  \omega'_1(z) \\ \gamma_2 \omega'_2(z) \end{array} \right]  \right\rangle  \right], \nonumber
\end{align}
and can be done easily. Indeed, one can minimize (\ref{eq.Nemirovski.subproblem.solution}) over $p$ and $\xi$ separately to obtain:
\begin{align} \label{eq.simplex.projection}
p_n = \max \left\{ 0, \exp \left( -1 - \mu - \frac{d_{2n} - \gamma_2 \omega'_{2n}(z)}{\gamma_2} \right) - \frac{\delta}{N} \right\}
\end{align}
where $\mu$ is the solution to the single-variable equation:
$$
\sum\limits_{n=1}^N \max \left\{ 0, \exp \left( -1 - \mu - \frac{d_{2n} - \gamma_2 \omega'_{2n}(z)}{\gamma_2} \right) - \frac{\delta}{N} \right\} = 1.
$$
For minimization over $\xi$ we have the following three cases corresponding to the different choices for $\omega_1(\cdot)$:
\begin{itemize}
\item $\mathcal{B}$ -- $\infty$-norm ball:
$$
\xi_k = \left\{ \begin{array}{ll} 
- 1 & \text{ if } \frac{ \gamma_1 \omega'_{1k}(z) - d_{1k} }{\gamma_1} < -1 \\
\frac{ \gamma_1 \omega'_{1k}(z) - d_{1k} }{\gamma_1} & \text{ if } -1 \leq \frac{ \gamma_1 \omega'_{1k}(z) - d_{1k} }{\gamma_1} \leq 1 \\
1 & \text{ if } \frac{ \gamma_1 \omega'_{1k}(z) - d_{1k} }{\gamma_1} > 1,
\end{array} \right.
$$
\item $\mathcal{B}$ -- 2-norm ball:
$$
\xi = \left\{ \begin{array}{ll} 
\frac{ \gamma_1 \omega'_{1}(z) - d_{1} }{\gamma_1} & \text{ if } \left\| \frac{ \gamma_1 \omega'_{1}(z) - d_1 }{\gamma_1} \right\|_2 \leq 1 \\
\frac{\gamma_1 \omega'_{1}(z) - d_{1}}{\left\| \gamma_1 \omega'_{1}(z) - d_1 \right\|_2} & \text{ otherwise,}
\end{array} \right.
$$
\item $\mathcal{B}$ -- simplex is analogous to (\ref{eq.simplex.projection}).
\end{itemize}
\subsection{The algorithm and its convergence} \label{subsec.prox.algorithm}
Now, we present the scheme of the algorithm. Each general step is illustrated with the specific case where $\mathcal{B}$ is the $\infty$-norm ball, $\omega_1(\xi) = \xi^\top \xi / 2$ and $\omega_2(p) = \sum\limits_{n=1}^N (p_n  + \delta / N) \log (p_n  + \delta / N)$.
\begin{enumerate}
\item Initialize with a point $(\xi^0,p^0 ) = z^0 \in Z = \mathcal{B} \times \Delta_N$. 
\item Given $z^{t-1}$, check if
$$
z^{t-1} = \text{Prox}_{z^{t-1}} (\Phi(z^{t-1})).
$$
\emph{In our example, this reduces to solving the problem:
\begin{align*}
w &= \argmin\limits_{\xi \in \mathcal{B}, p \in \Delta_N} \left[ \gamma_1 \frac{\xi^\top \xi}{2}  + \gamma_2 \sum\limits_{n=1}^N (p_n + \delta / N)\log (p_n + \delta / N) \right. \\ 
& \left. + \left\langle \left[ \begin{array}{c} \xi \\ p \end{array} \right], \Phi(z^{t-1})  - \left[ \begin{array}{c} \gamma_1  \xi \\ \gamma_2 (1 + \log (p + \delta / N)) \end{array} \right]  \right\rangle  \right]
\end{align*}
and checking if $w = z^{t-1}$.}

If equality holds, the algorithm stops and $z^{t-1}$ is the saddle point in (\ref{equation.saddle.crude}). Otherwise, go to Step 3.
\item Choose a $\gamma^t > 0$, set $w^{t,0} := z^{t-1}$ and do the inner iteration (over $k = 1,2,\ldots$):
\begin{align} \label{equation.inner.iteration}
w^{t,k} = \text{Prox}_{z^{t-1}} (\gamma^t \Phi(w^{t,k-1})).
\end{align}
until for the first time the following inequality holds:
\begin{align}
\left\langle \gamma^t \Phi(w^{t,k-1}), w^{t,k-1} - w^{t,k} \right\rangle + \omega(z^{t-1}) + \left\langle \omega'(z^{t-1}), w^{t,k} - z^{t-1} \right\rangle - \omega( w^{t,k-1} ) \leq 0. \label{equation.inner.iteration.condition}
\end{align}
\emph{In our example, this means solving the problem:
\begin{align*}
w &= \argmin\limits_{\xi \in \mathcal{B}, p \in \Delta_N} \left[ \gamma_1 \frac{\xi^\top \xi}{2}  + \gamma_2 \sum\limits_{n=1}^N (p_n + \delta / N)\log (p_n + \delta / N) \right. \\ 
& \left. + \left\langle \left[ \begin{array}{c} \xi \\ p \end{array} \right], \gamma^t \Phi(w^{t,k-1})  - \left[ \begin{array}{c} \gamma_1  \xi^{t,k-1} \\ \gamma_2 (1 + \log (p^{t,k-1} + \delta / N)) \end{array} \right]  \right\rangle  \right]
\end{align*}
until it holds that:
\begin{align*}
& \left\langle \gamma^t \Phi(w^{t,k-1}), w^{t,k-1} - w^{t,k} \right\rangle + \gamma_1 (\xi^{t-1})^\top \xi^{t-1} + \gamma_2 \sum\limits_{n=1}^N (p_n^{t-1} + \delta /N) \log (p_n^{t-1} + \delta /N) \\
& + \left\langle z^{t-1}, w^{t,k} - z^{t-1} \right\rangle - \gamma_1 (\xi^{t,k-1})^\top \xi^{t,k-1} - \gamma_2 \sum\limits_{n=1}^N (p_n^{t,k-1} + \delta /N) \log (p_n^{t,k-1} + \delta /N) \leq 0.
\end{align*}}
When the inner iteration is finished and condition (\ref{equation.inner.iteration.condition}) is satisfied, run the following updates:
\begin{align*}
w^t &= w^{t,k-1} \\
z^{t} &= w^{t,k} \\
y^t & =  \frac{\sum_{l=0}^t \gamma^l w^l}{\sum_{l=1}^t \gamma^l},
\end{align*}
where $y^t$ is the approximate solution to (\ref{equation.saddle.crude}), and go back to Step 1.
\end{enumerate}
The algorithm's complexity is stated in the following proposition based on Proposition 2.2 in  \cite{Nemirovski2004paper}.
\begin{proposition} \label{proposition.convergence}
Let $\alpha_1, \alpha_2 > 0$ be the strong convexity parameters of $\omega_1(\xi)$, $\omega_2(p)$, $\| \cdot \|_{(1)}$ and $\| \cdot \|_{(2)}$ be two norms and $L_{uv} > 0$, $u,v \in \{1,2\}$ such that:
$$
L_{uv} \geq \frac{\left\| \nabla_{z_u} \phi(z) - \nabla_{z_u} \phi(z') \right\|_{(u)*}}{\left\| z_v - z'_v \right\|_{(v)}}, \forall z, z' \in \mathcal{B} \times \Delta_N, \ z_v \neq z'_v, \ z_u = z'_u,
$$
and
$$
\Theta_1 \geq \sup\limits_{\xi, \xi' \in \mathcal{B}} \{\omega_1(\xi) - \omega_1(\xi') - \left\langle \xi - \xi', \omega_1'(\xi') \right\rangle \}, \quad \Theta_2 \geq \sup\limits_{p,p' \in \Delta_N} \{\omega_2(p) - \omega_2(p') - \left\langle p - p', \omega_2'(p') \right\rangle \},
$$
and $\gamma_1$, $\gamma_2$ in (\ref{equation.dgf}) be equal to:
$$
\gamma_1 = \frac{ \sum\limits_{l=1}^2 L_{1l} \sqrt{\frac{\Theta_1 \Theta_l}{\alpha_1 \alpha_l}} }{\Theta_1 \sum\limits_{k=1}^2 \sum\limits_{l=1}^2 L_{kl} \sqrt{\frac{\Theta_k \Theta_l}{\alpha_k \alpha_l}}}, \quad \gamma_2 = \frac{ \sum\limits_{l=1}^2 L_{2l} \sqrt{\frac{\Theta_2 \Theta_l}{\alpha_2 \alpha_l}} }{\Theta_2 \sum\limits_{k=1}^2 \sum\limits_{l=1}^2 L_{kl} \sqrt{\frac{\Theta_k \Theta_l}{\alpha_k \alpha_l}}},
$$
and $\gamma_t$ satisfy the following condition for all iterations $t$:
$$
\gamma^t \leq \frac{1}{\sqrt{2}  \sum\limits_{k=1}^2 \sum\limits_{l=1}^2 L_{kl} \sqrt{\frac{\Theta_k \Theta_l}{\alpha_k \alpha_l}} },
$$
Under these assumptions, in each iteration $t$ there are at most two inner iterations (\ref{equation.inner.iteration}) and the number $T$ of iterations to obtain an $\epsilon$-approximation of the solution to (\ref{equation.saddle.crude}) is
$$
\mathcal{O} \left( \left( \sum\limits_{k=1}^2 \sum\limits_{l=1}^2 L_{kl} \sqrt{\frac{\Theta_k \Theta_l}{\alpha_k \alpha_l}} \right) / \epsilon \right).
$$ \hfill $\square$
\end{proposition}
The cost of each iteration is $\mathcal{O}(NMS)$ -- it consists of computing the derivatives of $\phi$ w.r.t. $\xi$ and $p$ and computing the proximal operator at most a fixed number of times. The total complexity depends on the dependence of the terms $\Theta_1$, $\alpha_1$ on $S$. For example, if $\mathcal{B}$ is chosen to be an $S$-dimensional Euclidean ball, the complexity is:
$$
\mathcal{O} \left(  NMS \log N / \epsilon \right).
$$
\begin{remark}
In the experiments, we observe a strong dependence of the convergence speed on the parameters $\gamma^t$. The observed convergence is faster if $\gamma^t$ is larger. However, a small $\gamma^t$ is needed to guarantee that the number of inner iterations is at most 2. For that reason, we implemented an update rule multiplying $\gamma^t$ by 1.5 if the number of inner iterations was at most 2, and dividing it by 1.5 otherwise. This heuristic performed well in our computational study.
\end{remark}
\subsection{Average cost constraint} \label{subsec.avg.cost.constraint}
A practically relevant extension of (\ref{equation.cc}) is the situation where there is a constraint on the average cost of the input distribution \cite{Sutter2015}. In such case, the feasible set of $p$ is $\overline{\Delta}_N = \{p \in \mathbb{R}^N_{+}: \ 1^\top p = 1, \ a^\top b \leq b \}$ and the worst-case channel capacity is equal to:
\begin{align}
\min\limits_{Q \in \mathcal{Q}} \max\limits_{p \in \overline{\Delta}_N } \sum\limits_{n = 1}^N \sum\limits_{m = 1}^M p_n Q_{nm} \log \frac{Q_{nm}}{\sum\limits_{l=1}^N p_l Q_{lm}} \label{eq.avg.cost}
\end{align}
Our goal is to solve (\ref{eq.avg.cost}) efficiently with the same machinery as the basic algorithm. Hence, we want to transform (\ref{eq.avg.cost}) to a saddle point problem where the maximization and minimization are done over compact sets for which it is still easy to compute the proximal operators. To do this, we make the following mild assumption and recall a well-known property of the average mutual information.
\begin{assumption} \label{assumption.strict.feasibility}
The inner problem in (\ref{eq.avg.cost}) is strictly feasible and its holds that $\min_n a_n < b$ and $\max_n a_n > b$.\hfill $\square$
\end{assumption}
\begin{property}[See \cite{Cover2012}] \label{property.AMIM.bounds}
It holds that:
$$
0 \leq \sum\limits_{n = 1}^N \sum\limits_{m = 1}^M p_n Q_{nm} \log \frac{Q_{nm}}{\sum\limits_{l=1}^N p_l Q_{lm}} \leq \min \left\{ \log N, \log M \right\}.
$$
\end{property}
By dualizing the inner problem in (\ref{eq.avg.cost}) w.r.t. the average cost constraint, we obtain the following problem, which by Assumption~\ref{assumption.strict.feasibility} and strong duality is equivalent to (\ref{eq.avg.cost}):
\begin{align}
\inf\limits_{Q \in \mathcal{Q}} \inf\limits_{\lambda \geq 0} \sup\limits_{p \in \Delta_N} \sum\limits_{n = 1}^N \sum\limits_{m = 1}^M p_n Q_{nm} \log \frac{Q_{nm}}{\sum\limits_{l=1}^N p_l Q_{lm}} + \lambda \left( b - \sum\limits_{n=1}^N p_n a_n \right) \label{eq.avg.cost.saddle.formulation}
\end{align}
The expression in (\ref{eq.avg.cost.saddle.formulation}) is convex in $(Q,\lambda)$ and concave in $p$. However, the domain for $\lambda$ is unbounded whereas the use of the algorithm of \cite{Nemirovski2004paper} requires compactness of the domain in both $p$ and $(Q,\lambda)$. A way to ensure it is to find a $\Lambda$ such that $\lambda^* \leq \Lambda$. To find such a $\Lambda$, we fix some $Q \in \mathcal{Q}$ and consider the maximized function of $\lambda$ in (\ref{eq.avg.cost.saddle.formulation}): 
\begin{align}
\label{eq.g.function.def}
g(\lambda) = \sup\limits_{p \in \Delta_N} \sum\limits_{n=1}^N \sum\limits_{m=1}^M p_n Q_{nm} \log \frac{Q_{nm}}{\sum\limits_{l=1}^N p_l Q_{lm}} + \lambda \left( b - \sum\limits_{n=1}^N p_n a_n \right).
\end{align}
First, we will show that $g(\lambda)$ has a shape similar to the one in the left panel of Figure \ref{fig.g.lambda} (convex and diverging to $+\infty$ as $\lambda \rightarrow +\infty$). This will imply that there exists a point $\Lambda > 0$ such that $\partial g(\Lambda)$ contains a positive number. Thus, $g(\lambda)$ will be increasing for $\lambda > \Lambda$ and the minimizer $\lambda^*$ satisfies $\lambda^* \leq \Lambda$.
\begin{figure}
\centering
\captionsetup{font=scriptsize}
\caption{\textbf{Left panel:} Illustration of the general structure of the function $g(\lambda)$ of the saddle-point algorithm in Section~\ref{subsec.avg.cost.constraint}. The function is increasing to the right from $\lambda = 3.5$ because it is convex and the derivative at $\lambda = 3.5$ is positive. \textbf{Right panel:} Evolution of $p_1$, and $\xi$ over the algorithm iterations in the experiment of Section~\ref{subsec.BSC.experiment}.}
\label{fig.g.lambda}
\includegraphics[width=\textwidth]{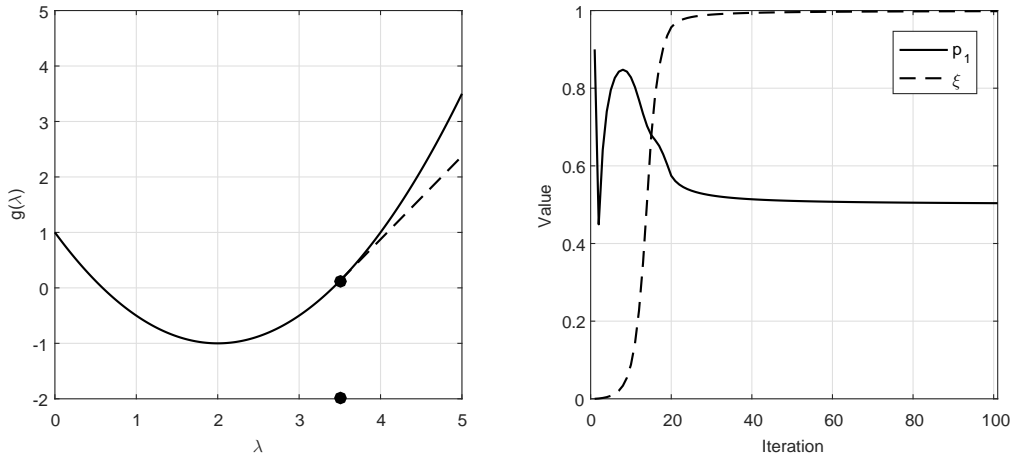}
\end{figure}

The function $g(\lambda)$ is convex as being a pointwise supremum of convex functions. For each $\lambda$ the supremum in (\ref{eq.g.function.def}) is attained because the maximized function of $p$ is continuous and $\Delta_N$ is compact. Next, we have the following lemma.
\begin{lemma} \label{lemma.g.lambda.infty}
Under Assumption~\ref{assumption.strict.feasibility} it holds that
$$
\lim\limits_{\lambda \rightarrow + \infty} g(\lambda)  = + \infty.
$$
\end{lemma}
\begin{proof}
See Appendix~\ref{appendix.proofs.5}.
\end{proof}

We know that $\min_{\lambda \geq 0} g(\lambda)$ is attained either at $\lambda^* = 0$ or at a $\lambda^* > 0$ such that $0 \in \partial g(\lambda^*)$. To find $\Lambda$ such that $g(\lambda)$ is increasing on $[\Lambda, + \infty)$ we have the following proposition.
\begin{proposition} \label{proposition.g.lambda.infty}
Under Assumption~\ref{assumption.strict.feasibility} and for $\Lambda > \log N / (b - \min_n a_n) $, there exists a $d > 0$ such that:
$$
d \in \partial g \left( \Lambda \right).
$$
\end{proposition}
\begin{proof}
See Appendix~\ref{appendix.proofs.6}.
\end{proof}

Since $d>0$ belongs to the subdifferential $\partial g(\Lambda)$ we obtain that for $t > 0$:
$$
g(\Lambda + t) \geq g(\Lambda) + td > g(\Lambda). 
$$
Hence, it holds that $\lambda^* \leq \Lambda = \log N / (b - \min_n a_n)$ and this bound is independent of $Q$. In this setup, the algorithm of \cite{Nemirovski2004paper} can be used to solve (\ref{eq.avg.cost.saddle.formulation}) -- it is only necessary to define a strongly convex distance generating function $\omega_1(\cdot)$ for the joint vector $(\xi,\lambda)$ over the set $\mathcal{B} \times [0, \Lambda]$. For example, when $\mathcal{B}$ is a $\infty$-norm box, such a function is $\omega_1(\xi,\lambda) = \xi^\top \xi /2 + \lambda^2 / 2$.
\section{Numerical experiments} \label{sec.experiments}
\subsection{Binary symmetric channel} \label{subsec.BSC.experiment}
Consider the binary symmetric channel as in Section~\ref{subsection.bsc}. We assume that the error probability belongs to an interval $\beta \in [0.15,0.45]$. In the setup of our algorithm this corresponds to the channel law matrix:
$$
Q(\xi) = Q_0 + \xi Q^1, \  Q_0 = \left[ \begin{array}{cc} 0.7 & 0.3 \\ 0.3 & 0.7  \end{array} \right],  Q^1 = \left[ \begin{array}{cc} -0.15 & 0.15 \\ 0.15 & -0.15  \end{array} \right], \ \xi \in [-1,1].
$$
We run the proximal algorithm with starting point $p^0 = (0.9, 0.1)^\top$ and $\xi^0 = 0$. By Proposition~\ref{proposition.bsc} we know that the saddle point corresponds to $p = (0.5, 0.5)^\top$ and $\xi = 1$. The right panel of Figure~\ref{fig.g.lambda} illustrates the evolution of the decision variables over the first 100 iterations.
\subsection{Impact of uncertainty} \label{subsubsec.impact.uncertainty}
The next example illustrates the impact of uncertainty on channel capacity. Consider the deterministic channel to be a randomly generated one, with a nominal channel law matrix $Q_0$ with $N = M = 50$ simulated randomly as:
$$
Q^0_{nm} = \frac{W_{nm}^4}{\sum\limits_{l=1}^M W_{nl}^4}, \quad  W_{nm} \sim \text{U}([1, 6.7]) \ \forall n,m.
$$
This way of sampling has two purposes: (i) keeping all the entries of $Q^0_{nm} \geq 10^{-5}$, (ii) reducing the number of `dominant' entries per row by using the 4-th power. 

We define the uncertainty as follows:
$$
Q(\xi) = Q^0 + \sum\limits_{s=1}^S \xi_s Q^{s}_{nm}, \quad \xi \in \mathcal{B},
$$
where
$$
Q^{s_n}_{nm} = \Gamma \left( \frac{1}{M} - Q^0_{nm} \right) \quad \mathcal{B} = \{\xi \in \mathbb{R}^S: \ 0 \leq \xi_s \leq 1, \ \forall s, \ \left\| \xi  \right\|_2 \leq 1 \},
$$
where $0 \leq \Gamma \leq 1$ and $s_n$, $n = 1,\ldots,N$ are sampled independently and uniformly from $\{1,2,\ldots,S\}$. In words, there are $S$ underlying primitive uncertainties $\xi_s$ and the ambiguity setup allows the perturbed channel to approach the channel in which every row of the channel law matrix is uniformly distributed, and the maximum size of perturbation is controlled by $\Gamma$. The Euclidean norm of the perturbation vector $\xi$ is restricted to be no greater than $1$, so that not all components $\xi_s$ can attain value 1 simultaneously. In the numerical setup we take $S = 5$.

Additionally, we consider the version of this problem with additional average cost constraint $a^\top p \leq b$. We set $b = 1$ and the $n$-component of $a$ is given by:
$$
a_n = \left\{ \begin{array}{cc} 50 & \text{ if } \overline{p}_n \geq 0.05 \\
 0 & \text{ otherwise,}  \end{array} \right.
$$
where $\overline{p}$ is the optimal solution to the problem without the constraint. In this way, we impose a very high penalty on the 10 input symbols whose optimal probability in the problem without the linear constraint was larger than or equal to 0.05.

The left panel in Figure~\ref{fig.uncertain.ellipsoid} illustrates the change in the channel capacity for the problem without and with the average cost constraint. We observe that the capacity of the channel drops significantly as the uncertainty scaling parameter $\Gamma$ increases. We see also that the capacity of the channel with the average cost constraint is strictly smaller than of the channel without such a constraint for all $\Gamma$.

\begin{figure}
\centering
\captionsetup{font=scriptsize}
\caption{\textbf{Left panel:} Worst-case channel capacities as a function of $\Gamma$ in the example of Section~\ref{subsubsec.impact.uncertainty} -- without and with the average cost constraint. \textbf{Right panel:} Worst-case channel capacities and the capacity for the deterministic channel (with $\xi = 0$) in the example of Section~\ref{subsubsec.single.uncertainty}. Capacities computed up to 0.01 accuracy.}
\label{fig.uncertain.ellipsoid}
\includegraphics[width=0.9\textwidth]{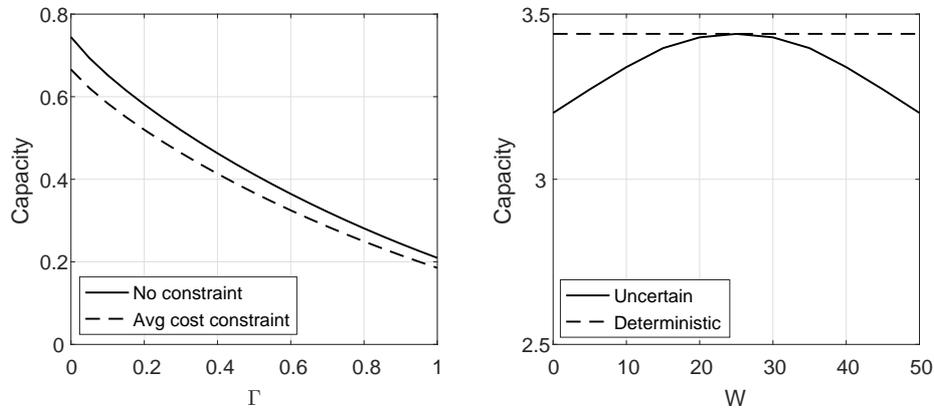}
\end{figure}
\subsection{Channel affected by a single uncertainty} \label{subsubsec.single.uncertainty}
Next we consider an experiment such that the channel is affected by a single large source of uncertainty and the idiosyncratic uncertainties of each row are negligible. 

We consider an uncertain channel with $N = M > 4$ where:
$$
Q(\xi) = Q^0 + \xi Q^p,
$$
where $\xi \in [-1,1]$ and the matrices are defined as:
{\scriptsize
$$
Q^0 = \left[ \begin{array}{ccccccc} 
0.92 - (m-5) \tau & 0.02 & 0.02 & 0.02 & 0.02 & \tau & \cdots \\
0.02 & 0.92 - (m-5) \tau & 0.02 & 0.02 & 0.02 & \tau & \cdots \\
0.02 & 0.02 &  0.92 - (m-5) \tau & 0.02 & 0.02  & \tau & \cdots \\
\\
\tau & \cdots & 0.02 & 0.02 &  0.92 - (m-5) \tau & 0.02 & 0.02  \\
\cdots & \tau & 0.02 & 0.02 & 0.02 & 0.92 - (m-5) \tau & 0.02 \\
\cdots & \tau & 0.02 & 0.02 & 0.02 & 0.02 &  0.92 - (m-5) \tau 
 \end{array} \right]
$$}
where $ \tau = 10^{-5}$ and the perturbation matrix is
{\scriptsize
$$
Q^p = \left[ \begin{array}{ccccccc} 
0.07 w_1 & 0.0175 w_1 & 0.0175 w_1 & 0.0175 w_1 & 0.0175 w_1 & 0 & \cdots \\
0.0175 w_2 & 0.07 w_2 & 0.0175 w_2 & 0.0175 w_2 & 0.0175 w_2 & 0 & \cdots \\
0.0175 w_3 & 0.0175 w_3 &  0.07 w_3 & 0.0175 w_3 & 0.0175 w_3  & 0 & \cdots \\
\\
\cdots & 0 & 0.0175 w_{m-2} & 0.0175 w_{m-2} &  0.07 w_{m-2} & 0.0175 w_{m-2} & 0.0175 w_{m-2}  \\
\cdots & 0 & 0.0175 w_{m-1} & 0.0175 w_{m-1} & 0.0175 w_{m-1} & 0.07 w_{m-1} & 0.0175 w_{m-1} \\
\cdots & 0 & 0.0175 w_{m} & 0.0175 w_{m} & 0.0175 w_{m} & 0.0175 w_{m} &  0.07 w_{m}
 \end{array} \right].
$$}
The nominal matrix $Q_0$ reflects the possible events following the choice of the $n$-th input symbol:
\begin{itemize}
\item the output symbol is the $n$-th output symbol with probability $0.92 - (N-5) \tau$;
\item the output symbol is one of the 4 nearest neighbors of the $n$-th output symbol, each with probability 0.02;
\item other output symbol occurs, each with a small probability $\tau$.
\end{itemize}
Perturbation of the channel takes away (or adds) the probability mass from the $n$-th symbol adding it to (taking it away from) the probability masses of these four neighboring symbols. We set the vector $w$ to be parametrized by integer $0 \leq W \leq m$ as follows:
$$
w = \left[ \underbrace{\begin{array}{ccc} -1 &\cdots & -1 \end{array}}_{W} \ \underbrace{\begin{array}{ccc} 1 &\cdots & 1 \end{array}}_{M - W} \right]^\top.
$$
In this way, the first $W$ rows of the channel transition matrix are affected by uncertainty $\xi$ in a way that is opposite to the remaining $M-W$ rows. We set $M = N = 50$. We run the algorithm up to accuracy 0.005. The right panel of Figure~\ref{fig.uncertain.ellipsoid} illustrates the results for different values of $W$.

As one can see, the figure is symmetric because the problem is symmetric around $W = 25$. For the extreme values of $W$, the capacity loss due to uncertainty amounts to $7\%$. We observe that for $W = 25$ the worst-case channel capacity is equal to the capacity of a channel defined by $Q_0$. This is because the worst-case realization of $\xi$ at the saddle point is exactly $\xi = 0$.
\section{Summary} \label{sec.summary}
In this paper we have derived the convex robust counterparts of the channel capacity problem affected by uncertainty in the channel law matrix, and presented a first-order saddle point algorithm with $\mathcal{O}(1/T)$ convergence rate that potentially can handle large-scale problems. Our numerical examples illustrate the fact that presence of uncertainty can lead to a decrease in channel capacity. 
\bibliographystyle{plainnat}

\begin{appendix}
\section{Proofs} \label{appendix.proofs}
\subsection{Proposition~\ref{proposition.dual.robust.counterpart}} \label{appendix.proofs.1}
\begin{proof}
Similar to the proof of Proposition 1, we use the results of \cite{BenTal2015} based on Fenchel duality. Denote the feasible set for $(t,p,V)$ by $\mathcal{D}$ and define
$$
g(t,p,V) = \sum\limits_{n=1}^N \lambda_n - \delta^\ast (V | \mathcal{Q})
$$
Then, we have:
\begin{align}
&\max_{\lambda_n, V_{nm}} \sum\limits_{n=1}^N \lambda_n - \delta^\ast (V | \mathcal{Q}) - \delta((t,p,V)|\mathcal{D}) \nonumber \\
= \ & \min\limits_{a,b,C} \delta^\ast((a,b,C)|\mathcal{D}) - g_\ast(a,b,C) \label{eq.dual.Fenchel.dual}
\end{align}
We first need to derive the support function $\delta^\ast((a,b,C)|\mathcal{D})$. We use the fact that $\mathcal{D}$ is described by the following convex functions:
\begin{align*}
h_m(\lambda,p,V) & \ = \sum\limits_{n=1}^N p_n \exp \left( \frac{\lambda_n - V_{nm}}{p_n}  \right) - 1 && m = 1,\ldots, M\\
h_{M+1}(\lambda,p,V) & \ = 1^\top p - 1 \\
h_{M+2}(\lambda,p,V) & \ = -1^\top p + 1.
\end{align*}
It holds that \cite{BenTal2015}:
\begin{align*}
\delta^\ast((a,b,C)|\mathcal{D}) = \inf\limits_{\substack{a^m,b^m, C^m, u_m \geq 0: \sum\limits_{m=1}^{M+2} a^m = a, \\ \sum\limits_{m=1}^{M+2} b^m = b, \\ \sum\limits_{m=1}^{M+2} C^m = C, }} & \left\{   \sum\limits_{m=1}^M u_m h_m^* \left( \frac{a^m}{u_m},\frac{b^m}{u_m},\frac{C^m}{u_m} \right) + u_{M+1} h_{M+1}^* \left( \frac{a^{M+1}}{u_{M+1}}, \frac{b^{M+1}}{u_{M+1}}, \frac{C^{M+1}}{u_{M+1}} \right) \right. \\
& + \left. u_{M+2} h_{M+2}^* \left( \frac{a^{M+2}}{u_{M+2}},\frac{b^{M+2}}{u_{M+2}},\frac{C^{M+2}}{u_{M+2}} \right) \right\}
\end{align*}
where $h_m^*(\lambda,p,V)$, $m =1,\ldots,M+2$ are the convex conjugates of the defining functions given by:
\begin{align*}
h_m^\ast(a,b,C) & \ = \left\{ \begin{array}{ll} 1 & \text{ if } \sum_{n=1}^N b_n + a_n \log a_n - a_n \leq 0 \ \forall j, \ a_n = - C_{nm} \ \forall j, \ C_{nl} = 0, \forall l \neq m \\
+\infty & \text{ otherwise.}
\end{array} \right. \\
h_{M+1}^\ast(a,b,C) & \ = \left\{ \begin{array}{ll} 1 & \text{ if } b \leq 1,\ a = 0, \ C = 0 \\
+\infty & \text{ otherwise.}
\end{array} \right. \\
h_{M+2}^\ast(a,b,C) & \ = \left\{ \begin{array}{ll} -1 & \text{ if } b \leq -1,\ a = 0, \ C = 0  \\
+\infty & \text{ otherwise.}
\end{array} \right. 
\end{align*}
Now, we need to derive the concave conjugate $g_\ast(a,b,C)$. Define:
\begin{align*}
g_1(\lambda,p,V) & \ =  \sum\limits_{n=1}^N \lambda_n \\
g_2(\lambda,p,V) & \ =  - \delta^\ast (V | \mathcal{Q}).
\end{align*}
Again, we know that \cite{BenTal2015}:
\begin{align*}
g_\ast(a,b,C) = \sup\limits_{\substack{a^1,b^1,C^1, a^2,b^2,C^2: \\ a^1 + a^2 = a,\ b^1 + b^2 = b,\ C^1 + C^2 = C}} \left\{g_{1\ast}(a^1,b^1,C^1) + g_{2\ast}(a^2,b^2,C^2) \right\}
\end{align*}
where:
\begin{align*}
g_{1\ast}(a,b,C) & \ = \left\{ \begin{array}{ll} 0 & \text{ if } a = 1, b = 0, \ C = 0  \\
-\infty & \text{ otherwise.}
\end{array} \right. \\
g_{2\ast}(a,b,C) & \ = \left\{ \begin{array}{ll} 0 & \text{ if } a = 0, b = 0, \ -C \in \mathcal{Q} \\
-\infty & \text{ otherwise.}
\end{array} \right. 
\end{align*}
Combining the two results -- for the support function and the concave conjugate -- we obtain the following problem equivalent to (\ref{eq.dual.Fenchel.dual}):
\begin{align*}
 \min_{a^m, b^m, C^m} \ & \sum_{m = 1}^M u_m + u_{M+1} - u_{M+2} \\
\text{s.t.} \ & \frac{b_n^m}{u_m} + \frac{a_n^m}{u_m} \log \frac{a_n^m}{u_m} - \frac{a_n^m}{u_m} \leq 0 && \forall n,m \\
& \frac{a_n^m}{u_m} = \frac{- C_{nm}^m}{u_m}  && \forall n,m \leq M \\
& C_{nl}^m = 0, &&\forall l \neq m \\
& u_m \geq 0, && \forall m \\
& \sum\limits_{m=1}^{M+2} a^m = 1, \ \sum\limits_{m=1}^{M+2} b^m = 0 \\
& \frac{b^{M+1}}{u_{M+1}} \leq 1, \ \frac{b^{M+2}}{u_{M+2}} \leq -1 \\
& C^{M+1} = C^{M+2} = 0, \ a^{M+1} = a^{M+2} = 0 \\
& - \sum\limits_{m=1}^{M+2} C^m \in \mathcal{Q}.
\end{align*}
To simplify, we multiply the first constraint by $u_m$, notice that $C^{M+1} = C^{M+2} = 0$, multiply the constraints in the 6-th row by $u_{M+1}$ and $u_{M+2}$, respectively, and substitute $a^m_n = Q_{nm}$ to obtain:
\begin{align}
 \min_{Q_{nm}, b^m} \ & \sum_{m=1}^M u_m + u_{M+1} - u_{M+2} \label{eq.appendix.support.problem} \\
\text{s.t.} \ & b_n^m + Q_{nm} \log \frac{Q_{nm}}{u_m} - Q_{nm} \leq 0 && \forall m \leq M, \ \forall n \nonumber \\
& u_m \geq 0, &&\forall m \nonumber \\
& b^{M+1} \leq u_{M+1} 1, \ b^{M+2} \leq -u_{M+2} 1 \nonumber \\
& \sum\limits_{m=1}^M b^m + b^{M+1} + b^{M+2} = 0 \nonumber \\
& Q \in \mathcal{Q}. \nonumber
\end{align}
Now we notice that:
$$
u_{M+1} - u_{M+2} \geq \max\limits_n \left\{ b^{M+1}_n + b^{M+2}_n \right\} \geq \max\limits_n \sum\limits_{m =1}^M - b^m_n \geq \max\limits_n \left\{ \sum\limits_{m=1}^M Q_{nm} \left( \log \frac{Q_{nm}}{u_m} - 1 \right) \right\},
$$
where the subsequent inequalities follow from the constraints in the 3rd, 4th and 1st rows of (\ref{eq.appendix.support.problem}), respectively. Thanks to this, we obtain the following problem equivalent to (\ref{eq.appendix.support.problem}):
\begin{align*}
\min_{Q_{nm}, u_m} \ & \sum_{m=1}^M u_m + \max\limits_n \left\{ \sum\limits_{m=1}^M Q_{nm} \log \frac{Q_{nm}}{u_m} - 1 \right\}
\end{align*}
The final formulation is obtained by substituting $u_m = \exp(v_m - \mu - 1)$ and minimizing w.r.t. $\mu$. Strong duality holds for the pair (\ref{eq.cc.rc}) - (\ref{eq.cc.rc.dual}) due to the fact that  Slater condition clearly holds for problem (\ref{eq.cc.rc}).
\end{proof}
\subsection{Proposition~\ref{proposition.weakly.symmetric.bound.tight}} \label{appendix.proofs.2}
\begin{proof}
It is enough to find values for the variables $\lambda$, $p$, $V$ in the primal problem (\ref{eq.cc.rc}) such that the objective in the primal problem attains the value in (\ref{eq.upper.bound.symmetric.min}). First, take $p_n = 1/N$ and $\lambda_n = (\log N) / N$ to obtain the following form of the problem:
\begin{align*}
\max_{V} & \log N - \delta^\ast (V | \mathcal{Q}) \\
\text{s.t.} & \sum\limits_{n=1}^N  \exp \left( - N V_{nm}  \right) \leq 1, \ \forall m.
\end{align*}
Set $V_{nm} = - (\log P_{nm})/N$ where $P^\top \in \Delta_{M \times N}$ (each column of $P$ belongs to a simplex). Then, the problem's constraints hold and for the objective we have:
\begin{align*}
\max_{P^\top \in \Delta_{n \times m}} \log N - \sup_{Q \in \mathcal{Q}} \sum_{n=1}^N \sum_{m=1}^M - \frac{1}{N} Q_{nm}(\log P_{nm}) \ & = \max_{P^\top \in \Delta_{n \times m}} \inf_{Q \in \mathcal{Q}} \log N + \sum_{m=1}^M \frac{1}{N} Q_{nm} \log P_{nm} \\
& = \max_{P^\top \in \Delta_{n \times m}} \inf_{Q \in \mathcal{Q}} \sum_{n=1}^N \sum_{m=1}^M \frac{1}{N} Q_{nm} \log \frac{P_{nm}}{1/N} \\
& = \inf_{Q \in \mathcal{Q}} \max_{P^\top \in \Delta_{n \times m}}  \sum_{n=1}^N \sum_{m=1}^M \frac{1}{N} Q_{nm} \log \frac{P_{nm}}{1/N} \\
& = \inf_{Q \in \mathcal{Q}} \sum_{n=1}^N \sum_{m=1}^M \frac{1}{N} Q_{nm} \log \frac{Q_{nm}}{(\sum_l Q_{lm})/N} \\
& = \inf_{Q \in \mathcal{Q}} \log N + \sum_{m=1}^M  Q_{1m} \log \frac{Q_{1m}}{\sum_l Q_{lm}},
\end{align*}
where the penultimate equality follows from the fact that the maximizer of 
$$
\max\limits_{P^\top \in \Delta_{N \times M}} \sum\limits_{n=1}^N \sum\limits_{m=1}^M  p_n Q_{nm} \log \frac{P_{nm}}{p_n}
$$
is given by (see \cite{BenTal1988})
$$
P_{nm} = \frac{Q_{nm} p_n}{ \sum\limits_{l=1}^N Q_{lm} p_l}
$$
and the last equality follows by symmetry of the channel -- each row is the same and the sum of entries in each column is the same.
\end{proof}
\subsection{Proposition~\ref{proposition.bsc}} \label{appendix.proofs.3}
\begin{proof}
To show that this is true, we need to give a set of primal and dual variables for the binary symmetric channel such that the objective function values of the primal and dual robust problems are equal.  This is achieved by the following set of values:
\begin{align*}
p_1^* = p_2^* = 1/2, \ \lambda_1^* = \lambda_2^* = \frac{\log 2}{2}, \ V_{11} = V_{22}^* = \frac{- \log(1 - \beta^*)}{2}, \ V_{12}^* = V_{21}^* = \frac{- \log \beta^*}{2}
\end{align*}
and
\begin{align*}
Q_{11}^* = Q_{22}^* = 1- \beta^*, \ Q_{12}^* = Q_{21}^* = \beta^*, \ u_1^* = u_2^* = 1/2.
\end{align*}
Inserting these values and noticing that:
\begin{align*}
\delta^*(V^*| \mathcal{Q}) \ & = \sup\limits_{\underline{\beta} \leq \beta \leq \overline{\beta}} \left\{ (1-\beta)(-\log(1-\beta^*)) + \beta(-\log \beta^*) \right\} \\
& = - \log(1-\beta^*) + \sup\limits_{\underline{\beta} \leq \beta \leq \overline{\beta}}  \beta \log \frac{1-\beta^*}{\beta*} \\
& = - \log(1-\beta^*) + \left\{ \begin{array}{ll} 0 & \text{ if } \overline{\beta} \geq 1/2 \\
\overline{\beta} \log \frac{1 - \overline{\beta}}{\overline{\beta}} & \text{ otherwise.} \end{array} \right.
\end{align*}
we conclude that the primal and dual objective values match indeed.
\end{proof}
\subsection{Proposition~\ref{proposition.Nemirovski.L.continuity}} \label{appendix.proofs.4}
\begin{proof}
Consider the entries of the Jacobian of $\Psi(Q,p)$ (Hessian of $f(Q,p)$). We have the following formulas, where on the right we write the upper bound on the absolute value of a given term:
\begin{align*}
\frac{\partial^2 f(Q,p)}{\partial Q_{nm} \partial p_{n^*}} & = \left\{ \begin{array}{llr} 
\log \frac{Q_{nm}}{\sum_{l=1}^N p_l Q_{lm}} + p_n (\sum_{l=1}^N p_l Q_{lm^*})  & \text{ if } n^* = n  &  \qquad 1- \log \tau \\
p_{n} (\sum_{l=1}^N p_l Q_{lm^*}) & \text{ if } n^* \neq n &  \qquad 1 \\
 \end{array}  \right. \\
\end{align*}
\begin{align*}
\frac{\partial^2 f(Q,p)}{\partial Q_{nm} \partial Q_{n^*m^*}} & = \left\{ \begin{array}{llr} 
\frac{p_n}{Q_{nm}} - \frac{p_n^2}{\sum_{l=1}^N p_l Q_{lm}}  & \text{ if } n^* = n, \ m^* = m &  \qquad \frac{2}{\tau} \\
 - \frac{p_n^2}{\sum_{l=1}^N p_l Q_{lm}}  & \text{ if } n^* = n, \ m^* \neq m &  \qquad \frac{1}{\tau}\\
 - \frac{p_n p_n^*}{\sum_{l=1}^N p_l Q_{lm}}  & \text{ if } n^* \neq n, \ m^* = m &  \qquad \frac{1}{\tau}\\
 0  & \text{ if } n^* \neq n, \ m^* \neq m &  \qquad 0 \\
 \end{array}  \right. \\
\end{align*}
\begin{align*}
\frac{\partial^2 f(Q,p)}{\partial p_{n} \partial p_{n^*}} & = \left\{ \begin{array}{llr} 
1 + \log Q_{jm^*}  - \log \left( \sum\limits_{l=1}^N p_l Q_{lm^*} \right) - \frac{p_{j} Q_{jm^*}}{\sum\limits_{l=1}^N p_l Q_{lm^*}}  & \text{ if } n^* = n &  \qquad 1 + \frac{1}{\tau} - \log \tau \\
- \frac{p_{j^*}Q_{nm}}{\sum_{l=1}^N p_l Q_{lm^*}} & \text{ if } n^* \neq n &  \qquad \frac{1}{\tau} \\
 \end{array}  \right. & 
\end{align*}
Due to the boundedness of the Jacobian, the mapping is Lipschitz continuous.
\end{proof}
\subsection{Lemma~\ref{lemma.g.lambda.infty}} \label{appendix.proofs.5}
\begin{proof}
Define a function $\underline{g}(\lambda)$:
$$
\underline{g}(\lambda) = \sum\limits_{n=1}^N \sum\limits_{m=1}^M \underline{p}_n Q_{nm} \log \frac{Q_{nm}}{\sum\limits_{l=1}^m \underline{p}_l Q_{lm}} + \lambda \left( b - \sum\limits_{n=1}^N \underline{p}_n a_n \right)
$$
where $\underline{p}_{n^*} = 1$ for $n^* = \min \{n: \ a_n = \min_n a_n \}$ and $\underline{p}_n = 0$ otherwise. It holds that $g(\lambda) \geq \underline{g}(\lambda)$ for all $\lambda \in \mathbb{R}$. Using Property~\ref{property.AMIM.bounds} we have:
$$
\underline{g}(\lambda) = \sum\limits_{n=1}^N \sum\limits_{m=1}^M \underline{p}_n Q_{nm} \log \frac{Q_{nm}}{\sum\limits_{l=1}^m \underline{p}_l Q_{lm}} + \lambda \left( b - \sum\limits_{j=1}^m \underline{p}_n a_n \right)  \geq 0 + \lambda \left( b - \min\limits_n a_n \right).
$$
By Assumption~\ref{assumption.strict.feasibility} we have $ b - \min_n a_n > 0$. In this way we obtain that $\underline{g}(\lambda)$ diverges to $+\infty$ as $\lambda \rightarrow + \infty$. Since $g(\lambda) \geq \underline{g}(\lambda)$, the claim follows.
\end{proof}
\subsection{Proposition~\ref{proposition.g.lambda.infty}} \label{appendix.proofs.6}
\begin{proof}
First, note that:
$$
\partial g(\lambda) \ni  \left\{ \left( b -  \sum\limits_{n=1}^N p_n a_n \right): p \in \arg\max\limits_{p \in \Delta_N} \sum\limits_{n=1}^N \sum\limits_{m=1}^M p_n Q_{nm} \log \frac{Q_{nm}}{\sum\limits_{l=1}^N p_l Q_{lm}} + \lambda \left( b - \sum\limits_{n=1}^N p_n a_n \right) \right\}
$$
By contradiction assume that for $\lambda > \log N / (b - \min_n a_n)$ it holds that all the maximizers $p$ are such that $b - \sum_{n=1}^N p_n a_n \leq 0$. For such such a $p$ we have the following bound on the value of $g(\lambda)$:
$$
g(\lambda) = \sum\limits_{n=1}^N \sum\limits_{m=1}^M p_n Q_{nm} \log \frac{Q_{nm}}{\sum\limits_{l=1}^N p_l Q_{lm}} + \lambda \left( b - \sum\limits_{n=1}^N p_n a_n \right) \leq \log N + 0 = \log N,
$$
where the last inequality follows from Property~\ref{property.AMIM.bounds}. Now, consider $\underline{p}_{n^*} = 1$ for $n^* = \min \{n: \ a_n = \min_n a_n \}$ and $\underline{p}_n = 0$ otherwise. In this case, we obtain:
\begin{align*}
g(\lambda) & = \sum\limits_{n=1}^N \sum\limits_{m=1}^M \underline{p}_n Q_{nm} \log \frac{Q_{nm}}{\sum\limits_{l=1}^m \underline{p}_l Q_{lm}} + \lambda \left( b - \sum\limits_{j=1}^m \underline{p}_n a_n \right)  \\
& \geq 0 + \lambda \left(b - \min_n a_n \right) \\
& > \frac{\log N}{b - \min\limits_n a_n} \left(b - \min_n a_n \right) \\
& = \log N,
\end{align*}
where the inequalities follow from Property~\ref{property.AMIM.bounds} and the assumption $\lambda > \log N / (b - \min_n a_n)$. We obtain a contradiction because there must be a maximizing $p$ such that $d = b - \sum_{n=1}^N p_n a_n > 0$.
\end{proof}
\section{Algorithm of \cite{Nemirovski2004paper} and its convergence} \label{appendix.Nemirovski}
In this Appendix we provide a self-contained version of the algorithm of \cite{Nemirovski2004paper} and the theorem stating its convergence. Consider the problem
\begin{align}
\min\limits_{x \in X} \max\limits_{y \in Y} \phi(x,y). \label{appendix.saddle.problem}
\end{align}
Define $z = (x,y)$, $Z = X\times Y$ in the Euclidean space and
\begin{align*}
\Phi(z) = \Phi(x,y) = \left[ \begin{array}{c} \frac{\partial \phi(x,y)}{\partial x} \\ - \frac{\partial \phi(x,y)}{\partial y} \end{array} \right].
\end{align*}
\begin{assumption}\label{appendix.assumption}
The sets $X$, $Y$ are convex and compact. It holds that $\phi(x,y) \in C^{1,1}$, i.e., $\phi(x,y)$ is differentiable and $\Phi(z)$ is Lipschitz continuous:
\begin{align*}
\left\| \Phi(z) - \Phi(z') \right\|_* \leq L \left\| z - z' \right\|,
\end{align*}
where $L > 0$ and $\| \cdot \|$ is a norm. There exist functions $\omega_1: X \rightarrow \mathbb{R}$, $\omega_2: Y \rightarrow \mathbb{R}$ that are $\alpha_1$ and $\alpha_2$-strongly convex:
\begin{align*}
\left\langle \omega_1'(x) - \omega_1'(x'), x - x' \right\rangle & \geq \alpha_1 \| x - x' \|_2^2 \quad \forall x,x' \in X \\
\left\langle \omega_2'(y) - \omega_2'(y'), y - y' \right\rangle & \geq \alpha_2 \| y - y' \|_2^2 \quad \forall y,y' \in Y
\end{align*} \hfill $\square$
\end{assumption}
Define 
\begin{align*}
\Theta_1 & = \max\limits_{x, x' \in X} \left\{ \omega_1(x) - \omega_1(x') - \langle x - x', \omega_1'(x') \rangle \right\} \\
\Theta_2 & = \max\limits_{y, y' \in Y} \left\{ \omega_2(y) - \omega_2(y') - \langle y - y', \omega_2'(y') \rangle \right\}. \\
\end{align*}
Let $\| \cdot \|_{(1)}$ and $\| \cdot \|_{(2)}$ be two norms and $L_{uv} > 0$, $u,v \in \{1,2\}$ be such that:
$$
L_{uv} \geq \frac{\left\| \nabla_{z_u} \phi(z) - \nabla_{z_u} \phi(z') \right\|_{(u)*}}{\left\| z_v - z'_v \right\|_{(v)}}, \forall z = (z_1,z_2) , \ z' = (z'_1,z'_2) \in \mathcal{B} \times \Delta_N, \ z_v \neq z'_v, \ z_u = z'_u.
$$
and define the distance generating function
$$
\omega(z) = \gamma_1 \omega_1 (x) + \gamma_2 \omega_2(y),
$$
where
$$
\gamma_1 = \frac{ \sum\limits_{l=1}^2 L_{1l} \sqrt{\frac{\Theta_1 \Theta_l}{\alpha_1 \alpha_l}} }{\Theta_1 \sum\limits_{k=1}^2 \sum\limits_{l=1}^2 L_{kl} \sqrt{\frac{\Theta_k \Theta_l}{\alpha_k \alpha_l}}}, \quad \gamma_2 = \frac{ \sum\limits_{l=1}^2 L_{2l} \sqrt{\frac{\Theta_2 \Theta_l}{\alpha_2 \alpha_l}} }{\Theta_2 \sum\limits_{k=1}^2 \sum\limits_{l=1}^2 L_{kl} \sqrt{\frac{\Theta_k \Theta_l}{\alpha_k \alpha_l}}}.
$$
In this setting, under a proper norm, the following is a Lipschitz constant $L$ of $\Phi(z)$ on $Z$:
$$
L = \sum\limits_{k=1}^2 \sum\limits_{l=1}^2 L_{kl} \sqrt{\frac{\Theta_k \Theta_l}{\alpha_k \alpha_l}},
$$
Define the proximal mapping:
$$
\text{Prox}_d(z) = \argmin\limits_{w \in Z} \left\{ \omega(w) + \langle \Phi(z) - \omega'(d),w \rangle   \right\}.
$$
In this setup, the algorithm is as follows.
\begin{enumerate}
\item Choose a starting point $z_0 \in Z$.
\item Given a $z^{t-1}$ check whether
$$
\text{Prox}_{z^{t-1}}(\Phi(z^{t-1})) = z^{t-1}.
$$
If this is the case, claim that $z^{t-1}$ is the solution to (\ref{appendix.saddle.problem}). If not, go to 3.
\item Choose $\gamma^t > 0$, set $w^{t,0} := z^{t-1}$ and run the iteration 
\begin{align*}
w^{t,s} = \text{Prox}_{z^{t-1}}(\gamma^t \Phi(w^{t,s-1}))
\end{align*}
until the condition
\begin{align*}
\langle \gamma^t \Phi(w^{t,s-1}), w^{t,s-1} - w^{t,s} \rangle + \omega(z^{t-1}) + \langle \omega'(z^{t-1}), w^{t,s} - z_{t-1} \rangle - \omega(w^{t,s}) \leq 0
\end{align*}
is met. Denote by $s_t$ the corresponding value of $s$ and set $z^t = w^{t,s_t-1}$ and $z_{t} = w^{t,s_t}$.
\end{enumerate}
In particular, when $\gamma^t \leq 1 / \sqrt{2} L$ the number of inner iterations is at most two. Convergence of the algorithm is stated in the following theorem, based on Proposition 2.2 in \cite{Nemirovski2004paper}.
\begin{theorem}
Assume Assumption~\ref{appendix.assumption} holds. Define
\begin{align*}
z^T = (x^T,y^T) = \frac{\sum\limits_{t=1}^T \gamma^t z^t}{ \sum\limits_{t=1}^T \gamma^t }
\end{align*}
Then, it holds that:
\begin{enumerate}
\item If the algorithm terminates at a certain step $T$ according to the rule in Step 1, then $z^{T-1}$ is a solution to (\ref{appendix.saddle.problem}).
\item If the algorithm does not terminate in the course of $T$ steps then 
$$
\max\limits_{y\in Y} \phi(x^T,y) - \min\limits_{x \in X} \phi(x,y^T) \leq \frac{1}{\sum\limits_{t=1}^T \gamma^t}.
$$
\end{enumerate}
\end{theorem}
\end{appendix}
\end{document}